\documentclass{amsart}
\usepackage{setspace}
\usepackage{fancyhdr}
\usepackage{amsmath}
\usepackage{amssymb}
\usepackage{amsthm}
\usepackage{epsfig}
\usepackage[mathcal,mathscr]{eucal}
\usepackage{verbatim}
\usepackage{mathrsfs}

\newtheorem{thm}{Theorem}[section]

\newtheorem{lem}[thm]{Lemma}
\newtheorem{prpn}[thm]{Proposition}

\numberwithin{equation}{section} 

\begin{document}
\title{Perpetual Cancellable American Call Option}
\author{Thomas J. Emmerling}
\address{Department of Mathematics\\
University of Michigan\\
Ann Arbor, MI 48109}
\email{the@umich.edu}
\thanks{The author is sincerely grateful to E. Bayraktar for valuable discussions which improved this paper.  Additionally, the author
would like to thank an anonymous referee and J. Detemple for their helpful suggestions and insightful comments.}


\onehalfspacing

\begin{abstract}
This paper examines the valuation of a generalized American-style
option known as a Game-style call option in an infinite time horizon
setting.  The specifications of this contract allow the writer to
terminate the call option at any point in time for a fixed penalty
amount paid directly to the holder. Valuation of a perpetual
Game-style put option was addressed by Kyprianou (2004) in a
Black-Scholes setting on a non-dividend paying asset. Here, we
undertake a similar analysis for the perpetual call option in the
presence of dividends and find qualitatively different explicit
representations for the value function depending on the relationship
between the interest rate and dividend yield.  Specifically, we find that the value function
is not convex when $r>d$.  Numerical results show the impact this phenomenon has upon 
the vega of the option.  
\end{abstract}

\maketitle

\section{Introduction}
In current times, it is not hard to imagine a financial system
burdened by illiquidity over a large cross section of total market
activity.  Under such circumstances, trading in the market
might cease to be an option even for large financial firms
interested in hedging their short contracts.  Indeed, cancelling or
recalling such contracts might be one of the few ways to effectively mitigate undesirable positions in turbulent times.  As such,
derivative securities which include callback provisions or
cancellable features represent attractive instruments to writers of
these contracts.  The following discussion addresses the valuation
of a common American-style claim with the aforementioned termination
specification built into the contract.  \\
\indent Kifer (2000) was the first to broach the problem of
valuation for American-style options with a cancellation feature
available to the short side of the contract. In that article, Kifer
applied the continuous time game theoretic results of Lepeltier and
Maingueneau (1984) to these generalized American options and found
the fair price was equal to the value of a Dynkin game (see e.g. Dynkin (1969), Neveu (1975)) between the
long and the short sides of the contract.  The close relationship between these options and Dynkin games fostered a renewed interest in such games and subsequently brought forth general existence and characterization results about the value of a Dynkin game (see e.g. Alvarez (2008), Ekstr\"{o}m (2006), Ekstr\"{o}m and Villeneuve (2006), Ekstr\"{o}m and Peskir (2008), Peskir (2008)).
With respect to game options,  Kuhn, Kyprianou, and van Schaik (2007) recently extended valuation results in a complete market framework to include more general payoffs than those considered in Kifer (2000).  Recent results in an incomplete market setting include Kuhn (2004), and Hamad\`{e}ne and Zhang (2008). \\
\indent Since
game-type derivatives are generalized American-style options, explicit solutions are rare in many settings. 
However, Kyprianou (2004) explicitly solved, under
the Black-Scholes framework, the valuation problem associated to a
particular game-type derivative known as the perpetual Israeli
$\delta$-penalty put option.  This analysis was limited to a put
option on a non-dividend paying asset following geometric brownian
motion. Within this framework, Kyprianou found that the strike price
was the only asset value for which optimal contract termination
would occur. This result for the put option is intuitive and, perhaps, suggests similar behavior
for its call option counterpart. Following that article, Kuhn and
Kyprianou (2007) addressed the finite expiry put option valuation
problem and found its explicit representation as a compound exotic
option.  In the following discussion, we consider the valuation
problem of a perpetual game call option on a dividend
paying asset. Recently, Kunita and Seko (2004) considered the finite
expiry version of this contract.  Here, we utilize some of the same
arguments while attempting to explicitly solve the valuation
problem.  In doing so, we find significant qualitative differences
with Kunita and Seko's finite expiry analysis and important
distinctions from the work done by Kyprianou (2004) on an infinite
expiry game put option with a
non-dividend paying asset.  Most recently, Alvarez (2009) explicitly characterized both the value and the optimal exercise policy of a 
minimum guaranteed payment game option when the underlying asset price 
follows a general linear, time homogeneous diffusion.  The payoff structure Alvarez (2009) considered is, indeed, very
similar to a game call option since the payoff of the former upon exercise by the holder is $\max(X-K_{1},0)+K_{2}$ where $K_{1}=K_{2}>0$.  Our analysis here is distinct from Alvarez (2009) since a regular call option payoff assumes $K_{1}\geq0$ and $K_{2}=0$.  We find that this slight
parameter difference significantly changes the solution to the optimal stopping problem even in the typical case when the underlying dynamics follow geometric brownian motion.

\indent The forthcoming discussion is organized as follows. Section
\ref{CanCallSetup} describes the economic setting and presents a few foundational valuation results. Section
\ref{val-dr} addresses the valuation problem when $r\leq d$.  Section \ref{val-rd} examines valuation when $r>d$.  Section \ref{CanNumRes}
presents results from a numerical approximation of the optimal
exercise and cancellation boundaries.  Section
\ref{CanConclusion} concludes the valuation discussion.  Section \ref{Appendix} elaborates on a few claims from prior sections.   

\section{Setup}
\label{CanCallSetup} The economic setting is the standard financial
market with constant coefficients.  We assume the underlying asset
process follows the geometric Brownian Motion process whose price
satisfies
\begin{eqnarray}
\begin{split}
\label{perCallGBM} \mathrm{d}X_{t}=(r-d)X_{t}\mathrm{d}t+\sigma
X_{t}\mathrm{d}W_{t}
\end{split}
\end{eqnarray}
where $r$ is the risk-free rate of interest assumed to be strictly
positive, $d$ is the dividend rate on the underlying asset assumed
to be non-negative, and $\sigma$ is the volatility of the asset's
return assumed to be strictly positive.  The dynamics in
(\ref{perCallGBM}) describe the risk-neutralized evolution of the
underlying asset process.  The process $W$ is a Brownian motion
under the risk-neutral measure $\mathbb{P}$.  \\
\indent Let $V^{*}(X_{t})$ denote the value at $t$ of a perpetual call
option with a cancellation feature available to the short side of
the contract with penalty $\delta$.  That is, the payoff to the
holder upon cancellation when $X_{t}=x$ is
$(x-K)^{+}+\delta$.  We will refer to this contract as a perpetual
$\delta$-penalty call option or simply a $\delta$-penalty call
option.  If the holder exercises with strategy $\sigma$ and the writer cancels with strategy $\tau$, then payoff 
to the holder of the contract is $Z_{\sigma,\tau}$ where
\begin{eqnarray}
Z_{s,t}:=(X_{s}-K)^{+}1_{\{s\leq t\}}+((X_{t}-K)^{+}+\delta)1_{\{t<s\}}
\end{eqnarray}

Please note that we denote both the volatility of the geometric brownian motion and the holder's exercise stopping time by $\sigma$.  In the sequel, it will be clear by the context as to which quantity $\sigma$ references.
Standard results (see e.g. Kyprianou (2004)) can be invoked to establish that the value
of the $\delta$-penalty call option is
\begin{equation}
\label{optstopping}
\begin{split}
V^{*}(x)&=\inf_{\tau\in\mathcal{S}_{0,\infty}}\sup_{\sigma\in
\mathcal{S}_{0,\infty}}\mathbb{E}[e^{-r(\sigma\wedge\tau)}\{((xN_{\tau}-K)^{+}+\delta)1_{\{\tau<\sigma\}}
+ (xN_{\sigma}-K)^{+}1_{\{\sigma\leq \tau\}}\}]\\
&=\sup_{\sigma\in\mathcal{S}_{0,\infty}}\inf_{\tau\in
\mathcal{S}_{0,\infty}}\mathbb{E}[e^{-r(\sigma\wedge\tau)}\{((xN_{\tau}-K)^{+}+\delta)1_{\{\tau<\sigma\}}
+ (xN_{\sigma}-K)^{+}1_{\{\sigma\leq \tau\}}\}]\\
\end{split}
\end{equation}
where 
\begin{equation}
N_{t}:=\exp{\{(r-d-\frac{\sigma^{2}}{2})t+\sigma W_{t}\}}
\end{equation}
with optimal exercise strategies for the holder and writer
respectively equal to
\begin{eqnarray}
\label{dopt-times}
\begin{split}
\sigma^{*}&=\inf{\{t\in[0,\infty):V^{*}_{t}=(X_{t}-K)^{+}\}}\\
\tau^{*}&=\inf{\{t\in[0,\infty):V^{*}_{t}=(X_{t}-K)^{+}+\delta\}}
\end{split}
\end{eqnarray}
where $\inf{\{\emptyset\}}=\infty$, by convention.  We shall adopt this convention throughout the entire paper.  Note $\mathcal{S}_{0,\infty}$
denotes the set of all stopping times of the Brownian filtration,
and $\mathbb{E}$ is the expectation under the risk-neutral measure
$\mathbb{P}$.  In addition, let $\mathbb{E}_{x}$ denote the
expectation under $\mathbb{P}$ such that $X_{0}=x$.\\
\indent We begin our discussion with a regularity result for the value
function of the $\delta$-penalty call option.  
\begin{prpn}
\label{perpLipschitz} The value function is non-decreasing in $x$
and is Lipschitz continuous with Lipschitz constant $1$.
\end{prpn}
\begin{proof}
Recall,
\begin{eqnarray*}
V^{*}(x)=\inf_{\tau\in\mathcal{S}_{0,\infty}}\sup_{\sigma\in
\mathcal{S}_{0,\infty}}\mathbb{E}[e^{-r(\sigma\wedge\tau)}\{((xN_{\tau}-K)^{+}+\delta)1_{\{\tau<\sigma\}}
+ (xN_{\sigma}-K)^{+}1_{\{\sigma\leq \tau\}}\}]
\end{eqnarray*}
Using the fact that $(x-K)^{+}$ is a non-decreasing function of $x$
and the definition
\begin{equation}
\begin{split}
J^{x}(\sigma,\tau)&:=\mathbb{E}[e^{-r(\sigma\wedge\tau)}\{((xN_{\tau}-K)^{+}+\delta)1_{\{\tau<\sigma\}}
+ (xN_{\sigma}-K)^{+}1_{\{\sigma\leq \tau\}}\}]\\
\end{split}
\end{equation}
we have $J^{x}(\sigma,\tau) \leq J^{y}(\sigma,\tau)$ for any
$\sigma$, $\tau\in \mathcal{S}_{0,\infty}$.  This implies $V^{*}(x)\leq
V^{*}(y)$ any $x<y$, i.e. $V$ is non-decreasing in $x$.  Now with the
following definitions
\begin{equation}
\begin{split}
\sigma_{x}&:=\inf{\{t\geq 0: V^{*}(xN_{t})=(xN_{t}-K)^{+}\}}\\
\tau_{y}&:=\inf{\{t\geq 0: V^{*}(yN_{t})=(yN_{t}-K)^{+}+\delta\}}
\end{split}
\end{equation}
and using the standard convention that $\inf{\{\emptyset\}}=\infty$,
we have for $x<y$
\begin{equation}
\begin{split}
V^{*}(y)&\leq J^{y}(\sigma_{x},\tau_{y})\\
V^{*}(x)&\geq J^{x}(\sigma_{x},\tau_{y})
\end{split}
\end{equation}
The following sequence of relations hold.

\begin{equation}
\begin{split}
V^{*}(y)-V^{*}(x)&\leq
J^{y}(\sigma_{x},\tau_{y})-J^{x}(\sigma_{x},\tau_{y})\\
&=\mathbb{E}[e^{-r(\sigma_{x}\wedge\tau_{y})}\{((yN_{\tau_{y}}-K)^{+}+\delta)1_{\{\tau_{y}<\sigma_{x}\}}
+ (yN_{\sigma_{x}}-K)^{+}1_{\{\sigma_{x}\leq \tau_{y}\}}\}]\\
& \quad
-\mathbb{E}[e^{-r(\sigma_{x}\wedge\tau_{y})}\{((xN_{\tau_{y}}-K)^{+}+\delta)1_{\{\tau_{y}<\sigma_{x}\}}
+ (xN_{\sigma_{x}}-K)^{+}1_{\{\sigma_{x}\leq \tau_{y}\}}\}]\\
&=\mathbb{E}[e^{-r(\sigma_{x}\wedge\tau_{y})}\{(yN_{\tau_{y}\wedge
\sigma_{x}}-K)^{+}-(xN_{\tau_{y} \wedge \sigma_{x}}-K)^{+}\}]\\
&\leq
\mathbb{E}[e^{-r(\sigma_{x}\wedge\tau_{y})}\{((y-x)(N_{\tau_{y}\wedge
\sigma_{x}})\}]\\
&=(y-x)\mathbb{E}[e^{-r(\sigma_{x}\wedge\tau_{y})}\{(N_{\tau_{y}
\wedge \sigma_{x}})\}] \\
&\leq y-x
\end{split}
\end{equation}
Note the final inequality holds since the discounted price of the
dividend paying asset is a $\mathbb{P}$-supermartingale.  Thus, $V^{*}$ is
Lipschitz continuous with Lipschitz constant $1$.  Note we have
shown, $0\leq V^{*}_{x}\leq 1$.
\end{proof}

The following notation will be utilized throughout the rest of the paper. Let
\begin{eqnarray}
\begin{split}
\lambda&:=\sqrt{2r+\left(\frac{r-d-\frac{\sigma^{2}}{2}}{\sigma}\right)^{2}}\\
\kappa &:=\frac{r-d-\frac{\sigma^{2}}{2}}{\sigma^{2}}\\
\end{split}
\end{eqnarray}
Our first valuation result identifies an upper bound on the penalty
for early cancellation.  More precisely, penalty values chosen above
this upper bound yield a $\delta$-penalty call option value exactly
equal to a perpetual call option since cancellation is not optimal.

\begin{prpn}
Let $v^{c}(x)$ denote the value of the perpetual call option on a
dividend paying asset at current level $x$ (see Section 2.6 Karatzas, Shreve (1998)).  Further, let
\begin{eqnarray}
\begin{split}
\delta^{*}:=v^{c}(K)=(b-K)\left(\frac{K}{b}\right)^{\frac{\lambda}{\sigma}-\kappa};
\ \text{where} \ b:=\frac{\frac{\lambda}{\sigma}-\kappa
}{\frac{\lambda}{\sigma}-\kappa-1}K.
\end{split}
\end{eqnarray}
If $\delta > \delta^{*}$, then the perpetual Israeli
$\delta$-penalty call option is precisely an American call option.
In other words, it is never optimal for the writer to cancel the
contract.
\end{prpn}

\begin{proof}
Suppose $\delta>\delta^{*}$.  Since $v^{c}(x)$ is an increasing
function of $x$ with derivative satisfying $0\leq v^{c}_{x}\leq 1$,
it follows that
\begin{eqnarray}
\begin{split}
(x-K)^{+}\leq v^{c}(x)\leq (x-K)^{+}+\delta
\end{split}
\end{eqnarray}
The following sequence of relations establishes the fact that the
$\delta$-penalty call option is simply an American call option. Note
$b$ denotes the optimal exercise boundary value for the American
call option and $\sigma_{x}:=\inf{\{t\geq 0: X_{t}=x\}}$.
\begin{equation}
\begin{split}
v^{c}(x)&= \inf_{\tau \in
\mathcal{S}_{0,\infty}}\mathbb{E}_{x}[e^{-r(\tau \wedge
\sigma_{b})}v^{c}(X_{\tau \wedge \sigma_{b}})]\\
&\leq \inf_{\tau \in \mathcal{S}_{0,\infty}}
\mathbb{E}_{x}[e^{-r(\tau \wedge \sigma_{b})}
\left((X_{\sigma_{{b}}}-K)^{+}1_{\{\sigma_{b}\leq
\tau\}}+((X_{\tau}-K)^{+}+\delta)1_{\{\tau<\sigma_{b}\}}\right)]\\
&\leq \sup_{\sigma \in \mathcal{S}_{0,\infty}}\inf_{\tau \in
\mathcal{S}_{0,\infty}} \mathbb{E}_{x}[e^{-r(\tau \wedge \sigma)}
\left((X_{\sigma}-K)^{+}1_{\{\sigma \leq
\tau\}}+((X_{\tau}-K)^{+}+\delta)1_{\{\tau<\sigma\}}\right)]\\
&\leq \sup_{\sigma \in \mathcal{S}_{0,\infty}}\mathbb{E}_{x}[e^{-r
\sigma}
(X_{\sigma}-K)^{+}]\\
&=v^{c}(x)
\end{split}
\end{equation}
The first equality follows since $v^{c}(x)$ is $r$-harmonic on $(0,b)$.  The first inequality follows since
$(s-K)^{+}\leq v^{c}(s)\leq (s-K)^{+}+\delta$ holds for all $s \in
(0,\infty)$. The second inequality follows by definition of the
supremum. The third inequality holds by definition of the infimum
and setting $\tau=\infty$. Note the order of the supremum and the
infimum in the second inequality can be reversed by starting from
the right-hand side and reasoning towards the left-hand side. Thus,
a saddle point occurs at $\sigma^{*}=\sigma_{b}$ and
$\tau^{*}=\infty$.
\end{proof}

\section{Valuation when $r \leq d$}
\label{val-dr} 
In this section, we wish to identify the value function of the $\delta$-penalty call option when the non-negative interest rate is bounded above by the dividend rate.  
The following theorem represents the main result of this section. 
\begin{thm}
Suppose $r\leq d$.  For $0<\delta \leq \delta^{*}$, the perpetual
$\delta$-penalty call option has value process $V(X_{t})$
where
\begin{equation}
\begin{split}
V(x)=
\begin{cases}
x-K & \textrm{if} \ x \in [k^{*},\infty)\\
 (k^{*}-K)\left(\frac{k^{*}}{x}\right)^{\kappa}\frac{\left(\frac{K}{x}\right)
 ^{-\frac{\lambda}{\sigma}}-\left(\frac{K}{x}\right)
 ^{\frac{\lambda}{\sigma}}}{\left(\frac{k^{*}}{K}\right)
 ^{\frac{\lambda}{\sigma}}-\left(\frac{k^{*}}{K}\right)
 ^{-\frac{\lambda}{\sigma}}}+\delta\left(\frac{K}{x}\right)^{\kappa}
 \frac{\left(\frac{k^{*}}{x}\right)
 ^{\frac{\lambda}{\sigma}}-\left(\frac{k^{*}}{x}\right)
 ^{-\frac{\lambda}{\sigma}}}{\left(\frac{k^{*}}{K}\right)
 ^{\frac{\lambda}{\sigma}}-\left(\frac{k^{*}}{K}\right)
 ^{-\frac{\lambda}{\sigma}}}
& \textrm{if} \ x \in (K,k^{*})\\
\delta \left(\frac{x}{K}\right)^{\frac{\lambda}{\sigma}-\kappa} & \textrm{if} \ x \in (0,K] \\
\end{cases}
\end{split}
\end{equation}
and the optimal exercise and cancellation strategies are
$\sigma^{*}:=\inf{\{t\geq 0: X_{t}\geq k^{*}\}}$ and
$\tau^{*}:=\inf{\{t\geq 0: X_{t}=K\}}$ where $k^{*}$ satisfies the
equation
\begin{eqnarray}
\begin{split}
&\left(\frac{k^{*}}{K}\right)^{\frac{2\lambda}{\sigma}}
\left(-2\left(\frac{K}{k^{*}}\right)^{\kappa+\frac{\lambda}{\sigma}}\delta
\lambda + (k^{*}-K)\left(\lambda-\kappa
\sigma+\left(\frac{K}{k^{*}}\right)^{\frac{2\lambda}{\sigma}}(\lambda+\kappa
\sigma)\right)\right)\\
& \quad
=k^{*}\left(-1+\left(\frac{k^{*}}{K}\right)^{\frac{2\lambda}{\sigma}}\right)\sigma.
\end{split}
\end{eqnarray}
\label{prop1}
\end{thm}

The proof of this theorem follows a path similar to the proof of
the value function for the perpetual $\delta$-penalty put option by
Kyprianou (2004).  In that paper, Kyprianou showed that the value
function for the put option is a convex function on $(0,\infty)$
when the penalty satisfies $\delta < v^{p}(K)$; where $v^{p}(K)$ is
the value function of a perpetual American put option on a
non-dividend paying asset when the asset price is equal to the
strike $K$.  When considering a call option on a dividend paying
asset with $r\leq d$, we find that the value function $V$ is also a
convex function on $(0,\infty)$ when the penalty satisfies
$0<\delta<v^{c}(K)$ (see Figure \ref{perCallDelt225RD01}).

\begin{figure}[h!!]
\centering
\epsfig{file=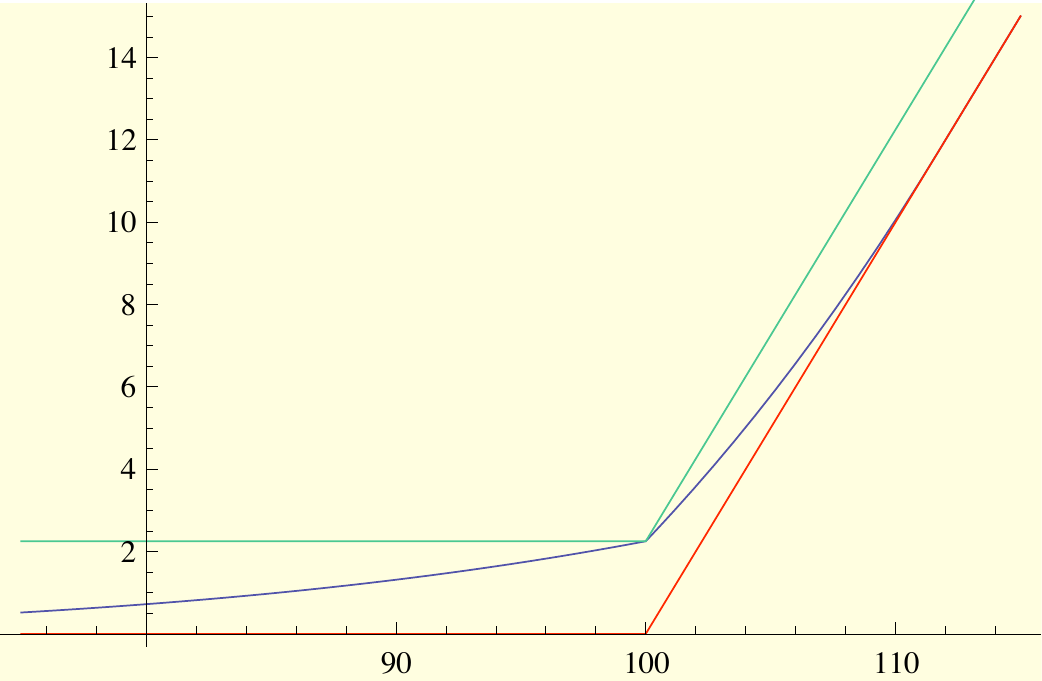,width=0.9\linewidth,clip=}
\caption[This figure displays the value function of the perpetual
$\delta$-penalty call option for $r\leq d$ along with the immediate
exercise and cancellation values.]{This figure displays the value
function of the perpetual $\delta$-penalty call option for $r\leq d$
along with the immediate exercise and cancellation values.  The
penalty is such that value function has increasing derivative at
$K$. Parameter values are $r=0.01, d =0.09$, $\sigma =0.20$,
$K=100$, $k^{*}=111.7641$,  and
$\delta=2.25$.}\label{perCallDelt225RD01}
\end{figure}

\begin{proof}
Suppose $0<\delta < \delta^{*}$.  We propose that the value function is
 $r$-harmonic on the set $(0,K) \cup (K,k^{*})$, satisfies the smooth fit condition at $k^{*}$ 
 and takes the value $\delta$ at  the strike price $K$.  Specifically, consider the boundary value problem
 
 \begin{eqnarray}
 \begin{split}
 \mathcal{L}v(x)=rv(x); \ v(K)=\delta, \ \lim_{x \downarrow 0}v(x)=0, \ \ x\in (0,K)
 \end{split}
 \end{eqnarray}
where $\mathcal{L}:= (r-d)x\frac{\mathrm{d}}{\mathrm{d} x}+\frac{1}{2}\sigma^{2}x^{2}\frac{\mathrm{d}^{2}}{\mathrm{d} x^{2}}$.  Let $v_{y}$ denote the derivative with respect to the parameter $y$.  Solving this problem yields, 

\begin{eqnarray}
\label{dr-delt}
\begin{split}
v(x)&=\delta \left(\frac{x}{K}\right)^{\frac{\lambda}{\sigma}-\kappa} \ \text{for} \ x\in(0,K)
\end{split}
\end{eqnarray}
Now consider the problem 

 \begin{eqnarray}
 \begin{split}
 \mathcal{L}v(x)=rv(x); \ v(K)=\delta, v(k^{*})=(k^{*}-K)^{+}, v_{x}(k^{*})=1, \  x\in (K,k^{*})
 \end{split}
 \end{eqnarray}
The solution to this problem is
\begin{equation}
\label{dr-other}
v(x)=
(k^{*}-K)\left(\frac{k^{*}}{x}\right)^{\kappa}\frac{\left(\frac{K}{x}\right)
 ^{-\frac{\lambda}{\sigma}}-\left(\frac{K}{x}\right)
 ^{\frac{\lambda}{\sigma}}}{\left(\frac{k^{*}}{K}\right)
 ^{\frac{\lambda}{\sigma}}-\left(\frac{k^{*}}{K}\right)
 ^{-\frac{\lambda}{\sigma}}}+\delta\left(\frac{K}{x}\right)^{\kappa}
 \frac{\left(\frac{k^{*}}{x}\right)
 ^{\frac{\lambda}{\sigma}}-\left(\frac{k^{*}}{x}\right)
 ^{-\frac{\lambda}{\sigma}}}{\left(\frac{k^{*}}{K}\right)
 ^{\frac{\lambda}{\sigma}}-\left(\frac{k^{*}}{K}\right)
 ^{-\frac{\lambda}{\sigma}}} \ \text{for} \ x\in(K,k^{*})
\end{equation}
where $k^{*}$ satisfies the following equation 
\begin{equation}
\label{FOC1}
\begin{split}
&\left(\frac{k^{*}}{K}\right)^{\frac{2\lambda}{\sigma}}
\left(-2\left(\frac{K}{k^{*}}\right)^{\kappa+\frac{\lambda}{\sigma}}\delta
\lambda + (k^{*}-K)\left(\lambda-\kappa
\sigma+\left(\frac{K}{k^{*}}\right)^{\frac{2\lambda}{\sigma}}(\lambda+\kappa
\sigma)\right)\right)\\
& \quad
=k^{*}\left(-1+\left(\frac{k^{*}}{K}\right)^{\frac{2\lambda}{\sigma}}\right)\sigma.
\end{split}
\end{equation}
\indent Simple calculations, using the fact that $r\leq d$, show that $v(x)$ is an increasing, convex function on $(0,K)$.  Before establishing that $v(x)$ is an increasing, convex function on $(K,k^{*})$, we first analyze its behavior at the strike price $K$.  The solution $v(x)$ of the boundary value problem is continuous but is not necessarily differentiable at $K$.  The following estimates show that the left-hand derivative is no larger than the right-hand derivative at $K$.  A non-decreasing derivative at $K$ requires 

\begin{eqnarray}
\begin{split}
\frac{2(k^{*}-K)\left(\frac{k^{*}}{K}\right)^{\kappa+\frac{\lambda}{\sigma}}\lambda
-2\left(\frac{k^{*}}{K}\right)^{\frac{2\lambda}{\sigma}}\delta \
\lambda}{\left(-1+\left(\frac{k^{*}}{K}\right)^{\frac{2\lambda}{\sigma}}\right)K
\ \sigma} \geq 0
\end{split}
\end{eqnarray}

Note the denominator is positive since $k^{*}\geq K$ and
$\frac{2\lambda}{\sigma}\geq 0$.  Hence, the derivative will be
increasing at $K$ if the following holds
\begin{eqnarray}
\begin{split}
(k^{*}-K)\left(\frac{k^{*}}{K}\right)^{\kappa+\frac{\lambda}{\sigma}}
-\left(\frac{k^{*}}{K}\right)^{\frac{2\lambda}{\sigma}}\delta &\geq
0 \ \Leftrightarrow \\
(k^{*}-K)\left(\frac{k^{*}}{K}\right)^{\kappa}
-\left(\frac{k^{*}}{K}\right)^{\frac{\lambda}{\sigma}}\delta & \geq
0
\end{split}
\end{eqnarray}
The left-side of this inequality is a decreasing, linear function of
$\delta$.  Thus, the condition on $\delta$ which guarantees the left-hand derivative is no larger than the right-hand derivative at $K$ is 
\begin{eqnarray}
\label{IncDeriv}
\label{deltacondition}
\begin{split}
\delta \leq
(k^{*}-K)\left(\frac{K}{k^{*}}\right)^{\frac{\lambda}{\sigma}-\kappa}
\end{split}
\end{eqnarray}
Interestingly, the assumption $0<\delta < \delta^{*}$ guarantees (\ref{IncDeriv}) holds.  One way in which to see this is to view $\delta$ as a function of $k^{*}$ in 
$(\ref{FOC1})$.  Indeed, the function $\delta(k^{*})$ is a continuous, increasing \footnote{See Section \ref{Appendix} for a justification of this claim} function such that $\delta(K)=0$ and $\delta(b)=v^{c}(K)$.  From this viewpoint and using this information, $(\ref{IncDeriv})$ will hold if
\begin{equation}
\label{fpositive}
\begin{split}
f(x)&:=(x-K)\left(\frac{K}{x}\right)^{\frac{\lambda}{\sigma}-\kappa}\\
&\quad \quad +\frac{\left(\frac{K}{x}\right)^{-\frac{\lambda+\kappa
\sigma}{\sigma}}\left(x\left(1-\left(\frac{x}{K}\right)^{\frac{-2\lambda}{\sigma}}
\right)\sigma-(s-K)\left(\lambda-\kappa\sigma+\left(\frac{K}{x}\right)^{\frac{2\lambda}{\sigma}
}(\lambda+\kappa\sigma)\right)\right)}{2\lambda}\\
&\geq 0, \ x\in [K,b]
\end{split}
\end{equation}

The function $f(x)$ is obtained by substituting the
representation for $\delta$ in terms of $k^{*}$ into
(\ref{deltacondition}) and then subtracting this term from each side
of the inequality.  Details of the proof that $f(x)\geq 0$ for $x\in [K,b]$ are included in Section \ref{Appendix}.  \\
\indent Continuing with the analysis of $v(x)$, its derivative on $(K,k^{*})$ (see formula $(\ref{dr-other})$) is 
\begin{eqnarray}
\label{big-form}
\begin{split}
&\left(\frac{1}{\left(-1+\left(\frac{k^{*}}{K}\right)^{\frac{2\lambda}{\sigma}}\right)x\sigma}\right)
\times \Bigg(\left(\frac{k^{*}}{K}\right)^{\frac{\lambda}{\sigma}}
\left(\frac{k^{*}}{x}\right)^{-\frac{\lambda}{\sigma}}
\left(\frac{K}{x}\right)^{-\frac{\lambda}{\sigma}} \\
& \quad \Bigg(-\left(\frac{K}{x}\right)^{\kappa+\frac{\lambda}{\sigma}}
\delta \left(\lambda-\kappa
\sigma+\left(\frac{k^{*}}{x}\right)^{\frac{2\lambda}{\sigma}}\left(\lambda+\kappa
\sigma \right)\right)\\
& \quad +
(k^{*}-K)\left(\frac{k^{*}}{x}\right)^{\kappa+\frac{\lambda}{\sigma}}
\left(\lambda-\kappa
\sigma+\left(\frac{K}{x}\right)^{\frac{2\lambda}{\sigma}}(\lambda+\kappa
\sigma)\right)\Bigg)\Bigg)
\end{split}
\end{eqnarray}
Note the first line in (\ref{big-form}) above has four factors that
are all positive.  Since $\delta \leq
(k^{*}-K)\left(\frac{K}{k^{*}}\right)^{\frac{\lambda}{\sigma}-\kappa}$, the expression in the remaining two lines of the above
derivative is greater than or equal to 
\begin{eqnarray}
\begin{split}
&(k^{*}-K)\Bigg(-\left(\frac{K}{k^{*}}\right)^{-\kappa+\frac{\lambda}{\sigma}}
\left(\frac{K}{x}\right)^{\kappa+\frac{\lambda}{\sigma}}
\left(\lambda-\kappa
\sigma+\left(\frac{k^{*}}{x}\right)^{\frac{2\lambda}{\sigma}}\left(\lambda+\kappa
\sigma \right)\right)\\
& \quad +\left(\frac{k^{*}}{x}\right)^{\kappa +
\frac{\lambda}{\sigma}}\left(\lambda-\kappa
\sigma+\left(\frac{K}{x}\right)^{\frac{2\lambda}{\sigma}}(\lambda+\kappa
\sigma)\right)\Bigg)
\end{split}
\end{eqnarray}
Since $k^{*}\geq K$, consider only the second factor in
the above representation.  Now taking a derivative yields
\begin{equation}
\begin{split}
\frac{2k^{*}\left(\frac{K}{k^{*}}\right)^{\frac{2\lambda}{\sigma}}
\left(\frac{k^{*}}{x}\right)^{-1+\frac{2\lambda}{\sigma}}\lambda(\lambda+\kappa\sigma)}{x^{2}\sigma}
-\frac{2K\left(\frac{K}{x}\right)^{-1+\frac{2\lambda}{\sigma}}\lambda(\lambda+\kappa\sigma)}{x^{2}\sigma}&=0
\end{split}
\end{equation}
Thus, the second factor is a constant function of $x$.  Substituting
the value at $x=k^{*}$ into the second factor produces

\begin{eqnarray}
\begin{split}
-\left(-1+\left(\frac{K}{k^{*}}\right)^{\frac{2\lambda}{\sigma}}(\lambda-\kappa
\sigma)\right)\geq 0
\end{split}
\end{eqnarray}
We conclude that $v(x)$ is increasing on $(K,k^{*})$.  Additionally, using the fact that $v(x)$ is $r$-harmonic on $(K,k^{*})$, $r\leq d$, $v_{x}(x)\geq 0$, and $v(x)>0$, we have
\begin{eqnarray}
\begin{split}
v_{xx}=\frac{2}{\sigma^{2}x^{2}}\left[(d-r)xv_{x}(x)+r v(x)\right]>0
\end{split}
\end{eqnarray}
Thus, $v(x)$ is convex on $(K,k^{*})$.  At this point, we conclude $v(x)$ is a convex function on $(0,\infty)$.  Summing up, $v(x)\in
\mathcal{C}^{2}(0,K)\cup\mathcal{C}^{1}(K,\infty)\cup
\mathcal{C}^{2}[(K,\infty)\setminus \{k^{*}\}]$, $v(x)$ is $r$-harmonic on $(0,K)\cup(K,k^{*})$ and $v(x)$ is $r$-superharmonic on $(k^{*},\infty)$.  Using these results, the following argument by Kyprianou (2004) proves that the solution to the boundary value problem is, indeed, the value function.  Let $\sigma_{k^{*}}:=\inf{\{t\geq 0: X_{t}\geq k^{*} \}}$ and $\tau_{K}:=\inf{\{t\geq 0: X_{t}=K\}}$. 
\begin{equation}
\begin{split}
v(x) &\leq \inf_{\tau \in
\mathcal{S}_{0,\infty}}\mathbb{E}_{x}[e^{-r(\tau \wedge
\sigma_{k^{*}})}v(X_{\tau \wedge \sigma_{k^{*}}})]\\
&\leq \inf_{\tau \in \mathcal{S}_{0,\infty}}
\mathbb{E}_{x}[e^{-r(\tau \wedge \sigma_{k^{*}})}
\left((X_{\sigma_{{k^{*}}}}-K)^{+}1_{\{\sigma_{k^{*}}\leq
\tau\}}+((X_{\tau}-K)^{+}+\delta)1_{\{\tau<\sigma_{k^{*}}\}}\right)]\\
&\leq \sup_{\sigma \in \mathcal{S}_{0,\infty}}\inf_{\tau \in
\mathcal{S}_{0,\infty}} \mathbb{E}_{x}[e^{-r(\tau \wedge \sigma)}
\left((X_{\sigma}-K)^{+}1_{\{\sigma \leq
\tau\}}+((X_{\tau}-K)^{+}+\delta)1_{\{\tau<\sigma\}}\right)]\\
&\leq \sup_{\sigma \in
\mathcal{S}_{0,\infty}}\mathbb{E}_{x}[e^{-r(\tau_{K} \wedge \sigma)}
\left((X_{\sigma}-K)^{+}1_{\{\sigma \leq
\tau_{K}\}}+((X_{\tau_{K}}-K)^{+}+\delta)1_{\{\tau_{K}<\sigma\}}\right)]\\
&\leq \sup_{\sigma \in \mathcal{S}_{0,\infty}}\mathbb{E}_{x}[e^{-r
(\tau_{K}\wedge \sigma)}v(X_{\tau_{K}\wedge \sigma})]\\
&\leq v(x)
\end{split}
\end{equation}
The first inequality follows since $v(x)$ is $r$-harmonic on $(0,K)\cup (K,k^{*})$. The
second inequality follows since $v(x)$ satisfies
$(x-K)^{+}\leq v(x) \leq (x-K)^{+}+\delta$.  The third and fourth
inequalities follow using the definition of the supremum and infimum
respectively. The fifth inequality holds using the same bound as in
the second inequality. The final inequality follows since $v(x)$ is $r$-superharmonic on $(k^{*},\infty)$.  Note, the order of the supremum and
infimum can be switched by establishing the above inequalities in
reverse.  This completes the proof. 
\end{proof}

\section{Valuation when $r>d$}
\label{val-rd}
Here we assume the interest rate $r$ is strictly larger than the
constant dividend yield $d$ of the underlying asset.  It seems reasonable to conjecture that the value function is identical to the solution found in the prior parameter case.  However, Figure \ref{perCallDelt10DR01} disproves this hypothesis since the proposed value function defined in Proposition \ref{prop1} does not
satisfy the basic inequality, 
\begin{equation}
(x-K)^{+}\leq V(x) \leq (x-K)^{+}+\delta
\end{equation}

\begin{figure}[h]
\centering
\epsfig{file=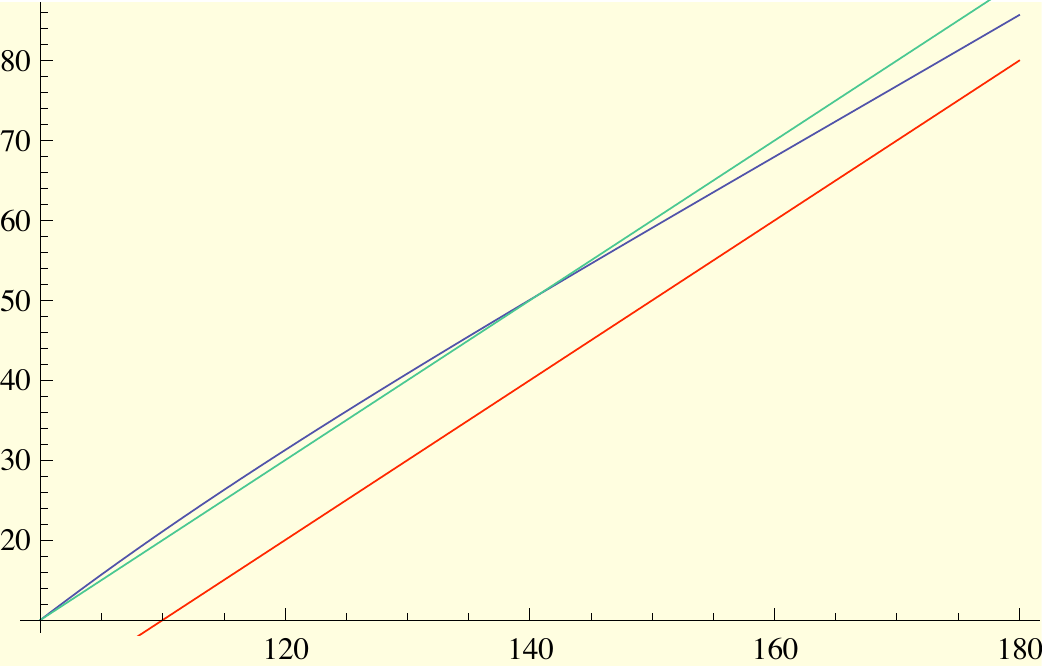,width=0.9\linewidth,clip=}
\caption[$\delta$-penalty call option:  This figure shows that the proposed 
value function (blue line) violates the upper bound on the value (green line) on the interval $(K,k^{*})$.]
{$\delta$-penalty call option:  This figure shows that the proposed 
value function (blue line) violates the upper bound on the value (green line) on the interval $(K,k^{*})$.  
Parameter values are $r=0.06, d =0.03$, $\sigma =0.20$,
$K=100$, $k^{*}=272.4404$, and
$\delta=10$.}\label{perCallDelt10DR01}
\end{figure}

With this information, it seems likely that the cancellation region is of the form $[K,h], \ h \neq K$.  The following argument suggests why the closed region
should be connected.  Suppose that for $x<y$,
$V(x)=(x-K)^{+}+\delta$ and $V(y)=(y-K)^{+}+\delta$ and that for
some $z$ where $x<z<y$, $V(z)<(z-K)^{+}+\delta$.  Since $V$
is a continuous function with derivative satisfying $0\leq
V_{x}\leq 1$ (see Proposition \ref{perpLipschitz}), we have an immediate contradiction. \\
\indent It is well-known that the fundamental solutions of the ordinary second order differential equation $\mathcal{L}v-rv=0$ are $\psi(x)=x^{\eta}$ and $\varphi(x)=x^{\nu}$, where

\begin{eqnarray}
\begin{split}
\eta&=\frac{1}{2}-\frac{r-d}{\sigma^{2}}+\sqrt{\left(\frac{1}{2}-\frac{r-d}{\sigma^{2}}\right)^{2}+\frac{2r}{\sigma^{2}}}>0\\
\nu&= \frac{1}{2}-\frac{r-d}{\sigma^{2}}-\sqrt{\left(\frac{1}{2}-\frac{r-d}{\sigma^{2}}\right)^{2}+\frac{2r}{\sigma^{2}}}<0
\end{split}
\end{eqnarray}

are the roots of the equation $\sigma^{2}(y-1)y+2(r-d)=2r$.  In addition, $\psi'(x)-\varphi(x)-\varphi'(x)\psi(x)=B S'(x)$.  Here, $B>0$ denotes the Wronskian of the fundamental solutions $(\psi(x)$, $\varphi(x))$, and $S'(x)$ is the density of the scale function $S$, where 
\begin{eqnarray}
\begin{split}
S(x):=\int_{c}^{x}\exp\left(-2\int_{c}^{y}\frac{(r-d)z}{\sigma^{2}z^{2}}\mathrm{d}z\right)\mathrm{d}y, \ \text{for} \ x \in (0,\infty)
\end{split}
\end{eqnarray}
where $c$ is an arbitrary fixed element of $(0,\infty)$.  The functions 
\begin{eqnarray}
\begin{split}
\hat{\psi}_{h}(x):=\psi(x)-\frac{\psi(h)}{\varphi(h)}\varphi(x)\\
\hat{\varphi}_{k}(x):=\varphi(x)-\frac{\varphi(k)}{\psi(k)}\psi(x)
\end{split}
\end{eqnarray}
are the fundamental solutions of $\mathcal{L}v-rv=0$ defined on the domain of the differential operator of the killed diffusion $\{X_{t}:t \in [0, \lambda_{h}\wedge \lambda_{k})\}$; $\lambda_{a}:=\inf\{t \geq 0: X_{t}=a\}$.  Finally, the density of the speed measure of $X_{t}$ is  $m'(y)=\frac{2}{\sigma^{2}y^{2}S'(y)}$ (see Borodin, Salminen (1996) Chapter 2 for details).  \\
\indent Using the above information, we now present the main result of this section.  Let $\sigma_{k^{*}}:=\inf{\{t\geq 0: X_{t}\geq k^{*}\}}$, and $\tau_{[K,h^{*}]}:=\inf{\{t\geq 0: K \leq X_{t}\leq h^{*}\}}$ where $k^{*}$ and $h^{*}$ are defined below.

\begin{thm}
\label{thm-rd}  Suppose $r\geq d$.  For $0<\delta \leq \delta^{*}$, the perpetual $\delta$-penalty call option has value process $V(X_{t})$ with
\begin{equation}
\label{valuefunction}
\begin{split}
V(x)=
\begin{cases}
x-K & \textrm{if} \ x \in [k^{*},\infty)\\
(k^{*}-K)^{+}\mathbb{E}_{x}[e^{-r\sigma_{k^{*}}}1_{\{\sigma_{k^{*}}\leq
\tau_{[K,h^{*}]}\}}] \\
\quad + ((h^{*}-K)^{+}+\delta)
\mathbb{E}_{x}[e^{-r\tau_{[K,h^{*}]}}1_{\{\tau_{[K,h^{*}]}<\sigma_{k^{*}}\}}]
& \textrm{if} \ x \in (h^{*},k^{*})\\
(x-K)+\delta & \textrm{if} \ x \in [K,h^{*}]\\
\delta \ \mathbb{E}_{x}[e^{-r\tau_{[K,h^{*}]}}] & \textrm{if} \ x \in (0,K)
\end{cases}
\end{split}
\end{equation}
where the pair $(h^{*}, k^{*})$ both satisfies the equations

\begin{eqnarray}
\label{opt-hk}
\begin{split}
\frac{1}{S'(h^{*})}\hat{\varphi}_{k^{*}}(h^{*})-\frac{\hat{\varphi}'_{k^{*}}(h^{*})}{S'(h^{*})}((h^{*}-K)^{+}+\delta)&=B\frac{(k^{*}-K)^{+}}{\psi(k^{*})}\\
\frac{1}{S'(k^{*})}\hat{\psi}_{h^{*}}(k^{*})-\frac{\hat{\psi}'_{h^{*}}(k^{*})}{S'(k^{*})}(k^{*}-K)^{+}&=-B\frac{(h^{*}-K)^{+}+\delta}{\varphi(h^{*})}
\end{split}
\end{eqnarray}
and the inequalities $K<h^{*}<k^{*}$.  Thus, the value function $V$ is continuous for all $x>0$ and is differentiable at $h^{*}$ and $k^{*}$ (by (\ref{opt-hk})).    
\end{thm}

The distinctive feature of this valuation formula is that the writer's termination region is the interval
$[K,h^{*}]$ for $h^{*}>K$ rather than simply the singleton $\{K\}$.
 Intuition for this result arises by examining the instantaneous gain
 to the writer for terminating the contract at time $t$.  A positive value for $rK-dX_{t}-\delta$
 provides an incentive for the writer to terminate the call option.  This may occur when
 the interest rate $r$ is larger than the dividend rate $d$.  If
 such a situation develops, then immediate termination by the writer might be
 preferable for some asset values strictly greater than the strike price
 (e.g. see Figure \ref{PerCallRD01}).  Before proving Theorem \ref{thm-rd}, we state a useful lemma concerning the pair $(h^{*},k^{*})$ whose proof appears in Section \ref{Appendix}.
 
 \begin{lem}
 \label{lemma-hk}
 A pair $(h^{*},k^{*})$ solving the equations (\ref{opt-hk}) with $K<h^{*}<k^{*}$ satisfies $h^{*}<\frac{r(K-\delta)}{d}$ and $k^{*}>\frac{r}{d}K$.
 \end{lem}
We now prove the main result.
\begin{proof}
\textit{(of Theorem \ref{thm-rd}}) 
Recall, $V^{*}(x)$ denotes the value function from (\ref{optstopping}).
Here, we intend to show $V^{*}(x)=V(x)$ for $x>0$ by establishing the following sequence of relations
\begin{equation}
\label{pfstrat}
\begin{split}
V(x)&\geq\sup_{\sigma\in
\mathcal{S}_{0,\infty}}\mathbb{E}_{x}[e^{-r(\sigma \wedge
\tau_{[K,h^{*}]})}\{((X_{\tau_{[K,h^{*}]}}-K)^{+}+\delta)1_{\{\tau_{[K,h^{*}]}<\sigma\}}
+(X_{\sigma}-K)^{+}1_{\{\sigma\leq
\tau_{[K,h^{*}]}\}}\}]\\
&\geq \inf_{\tau\in \mathcal{S}_{0,\infty}}\sup_{\sigma\in
\mathcal{S}_{0,\infty}}\mathbb{E}_{x}[e^{-r(\sigma \wedge
\tau)}\{((X_{\tau}-K)^{+}+\delta)1_{\{\tau<\sigma\}}
+(X_{\sigma}-K)^{+}1_{\{\sigma\leq
\tau\}}\}]\\
&\geq\sup_{\sigma\in \mathcal{S}_{0,\infty}}\inf_{\tau\in
\mathcal{S}_{0,\infty}}\mathbb{E}_{x}[e^{-r(\sigma \wedge
\tau)}\{((X_{\tau}-K)^{+}+\delta)1_{\{\tau<\sigma\}}
+(X_{\sigma}-K)^{+}1_{\{\sigma\leq
\tau\}}\}]\\
&\geq V(x)
\end{split}
\end{equation}

Notice that justification of the the first and last relations will complete the proof.  We begin by establishing the first inequality.  By (\ref{valuefunction}) and (\ref{opt-hk}), $V$ is continuously differentiable everywhere except at $K$, and twice continuously differentiable everywhere except at $K$, $h^{*}$, and $k^{*}$.  Using the change-of-variable formula with local time on curves (Peskir (2005) Remark 2.3) applied to $e^{-rt}V(X_{t})$, we obtain

\begin{eqnarray}
\label{peskir}
\begin{split}
e^{-rt}V(X_{t})&=V(x)+\int_{0}^{t}(\mathcal{L}V-rV)(s,X_{s})1_{\{X_{s}\neq
k^{*}\}\cap \{X_{s}\neq
h^{*}\}\cap \{X_{s}\neq K\}}\mathrm{d}s\\
& \quad \quad +\int_{0}^{t}e^{-rs}\sigma
X_{s}V_{x}(X_{s})1_{\{X_{s}\neq k^{*}\}\cap \{X_{s}\neq h^{*}\}\cap
\{X_{s}\neq K\}}\mathrm{d}W_{s}\\
& \quad \quad \quad
+\frac{1}{2}\int_{0}^{t}e^{-rs}(V_{x}(X_{s}+)-V_{x}(X_{s}-))1_{\{X_{s}=K\}}\mathrm{d}\ell^{K}_{s}(X)
\end{split}
\end{eqnarray}
where $\ell^{c}_{s}(X)$ is the local time of $X$ at the curve $c$
given by
\begin{eqnarray}
\begin{split}
\ell^{c}_{s}(X)=\lim_{\epsilon\downarrow 0}
\frac{1}{2\epsilon}\int^{s}_{0}1_{\{c(v)-\epsilon<X_{v}<c(v)+\epsilon\}}\mathrm{d}[X]_{v}
\end{split}
\end{eqnarray}

In the following, let $(\tau_{n})_{n=1}^{\infty}$ be a localizing sequence for the continuous local martingale, 
\begin{eqnarray}
\int_{0}^{t}e^{-rs}\sigma
X_{s}V_{x}(X_{s})1_{\{X_{s}\neq k^{*}\}\cap \{X_{s}\neq h^{*}\}\cap
\{X_{s}\neq K\}}\mathrm{d}W_{s}
\end{eqnarray}

\indent Let $x\in(h^{*},k^{*})$. Using the fact that $\mathcal{L}V=rV$ in
$(h^{*},k^{*})$ and the optional sampling theorem, we know for each
$n\geq 1$,
\begin{eqnarray}
\begin{split}
\mathbb{E}_{x}[e^{-r(\tau_{[K,h^{*}]}\wedge
\sigma_{k^{*}}\wedge \tau_{n})}V(X_{\tau_{[K,h^{*}]}\wedge
\sigma_{k^{*}}\wedge \tau_{n}})]=V(x)
\end{split}
\end{eqnarray}
Letting $n\rightarrow \infty$, we have by the bounded convergence theorem and the continuity of $V$,
\begin{eqnarray}
\label{mart-01}
\begin{split}
\mathbb{E}_{x}[e^{-r(\tau_{[K,h^{*}]}\wedge
\sigma_{k^{*}})}V(X_{\tau_{[K,h^{*}]}\wedge
\sigma_{k^{*}}})]=V(x)
\end{split}
\end{eqnarray}
This same argument also shows that (\ref{mart-01}) holds for $x\in(0,K)$ since $\mathcal{L}V=rV$ there.  Since (\ref{mart-01}) clearly holds when $x\in[K,h^{*}]$ and when $x\in[k^{*},\infty)$, we conclude (\ref{mart-01}) holds for all $x>0$.  Using Lemma \ref{lemma-hk}, we know for any $x \in [k^{*},\infty)$, 
\begin{equation}
(\mathcal{L}g_{1}-r g_{1})(x)=(r-d)x-r(x-K)=rK-dx<0
\end{equation}
where $g_{1}(x):=(x-K)^{+}$.  Therefore, for $x\in (h^{*},k^{*})$ and any $n\geq 1$,
\begin{eqnarray}
\begin{split}
\mathbb{E}_{x}[e^{-r(\sigma\wedge\tau_{[K,h^{*}]}\wedge
\tau_{n})}V(X_{\sigma\wedge\tau_{[K,h^{*}]}\wedge \tau_{n}})]\leq V(x)
\end{split}
\end{eqnarray}
Thus, by Fatou's lemma
\begin{eqnarray}
\begin{split}
\mathbb{E}_{x}[e^{-r(\sigma\wedge\tau_{[K,h^{*}]})}V(X_{\sigma\wedge\tau_{[K,h^{*}]}})]\leq
V(x)
\end{split}
\end{eqnarray}
Using Lemma \ref{lem-ineq}, we find
\begin{equation}
\begin{split}
&\mathbb{E}_{x}[e^{-r(\sigma \wedge
\tau_{[K,h^{*}]})}\{((X_{\tau_{[K,h^{*}]}}-K)^{+}+\delta)1_{\{\tau_{[K,h^{*}]}<\sigma\}}
+(X_{\sigma}-K)^{+}1_{\{\sigma\leq
\tau_{[K,h^{*}]}\}}\}]\\
&\leq \mathbb{E}_{x}[e^{-r(\sigma\wedge\tau_{[K,h^{*}]})}V(X_{\sigma\wedge\tau_{[K,h^{*}]}})]\\
&\leq V(x)
\end{split}
\end{equation}
Taking the supremum over all stopping times $\sigma$ yields,
\begin{equation}
\begin{split}
&\sup_{\sigma\in \mathcal{S}_{0,\infty}}\mathbb{E}_{x}[e^{-r(\sigma
\wedge
\tau_{[K,h^{*}]})}\{((X_{\tau_{[K,h^{*}]}}-K)^{+}+\delta)1_{\{\tau_{[K,h^{*}]}<\sigma\}}
+(X_{\sigma}-K)^{+}1_{\{\sigma\leq
\tau_{[K,h^{*}]}\}}\}]\\
& \quad \leq V(x)
\end{split}
\end{equation}
Thus, the first inequality of (\ref{pfstrat}) holds when
$x\in(h^{*},k^{*})$.  Continuing when $x\in(h^{*},k^{*})$, recall 
\begin{equation}
\begin{split}
&\mathbb{E}_{x}[e^{-r(\tau_{[K,h^{*}]}\wedge
\sigma_{k^{*}})}V(X_{\tau_{[K,h^{*}]}\wedge
\sigma_{k^{*}}})]\\
& \quad =\mathbb{E}_{x}[e^{-r(\sigma_{k^{*}} \wedge
\tau_{[K,h^{*}]})}\{((X_{\tau_{[K,h^{*}]}}-K)^{+}+\delta)1_{\{\tau_{[K,h^{*}]}<\sigma_{k^{*}}\}}
+(X_{\sigma_{k^{*}}}-K)^{+}1_{\{\sigma_{k^{*}}\leq
\tau_{[K,h^{*}]}\}}\}]\\
& \quad =V(x)
\end{split}
\end{equation}
Thus,
\begin{equation}
\begin{split}
V(x) &\geq
\inf_{\tau\in\mathcal{S}_{0,\infty}}\mathbb{E}_{x}[e^{-r(\sigma_{k^{*}}
\wedge \tau)}\{((X_{\tau}-K)^{+}+\delta)1_{\{\tau<\sigma_{k^{*}}\}}
+(X_{\sigma_{k^{*}}}-K)^{+}1_{\{\sigma_{k^{*}}\leq \tau\}}\}]
\end{split}
\end{equation}
We now establish the opposite inequality.  Using Lemma \ref{lemma-hk}, for any $x\in \left(K,\frac{r(K-\delta)}{d}\right)$, 
\begin{equation}
(\mathcal{L}g_{2}-r g_{2})(x)=(r-d)x-r((x-K)+\delta)=rK-dx-r \delta>0
\end{equation}
where $g_{2}(x):=(x-K)^{+}+\delta$.  Therefore, for $x\in (h^{*},k^{*})$ and $n\geq 1$ the optional
sampling theorem yields for any $\tau\in \mathcal{S}_{0,\infty} $,
\begin{eqnarray}
\begin{split}
\mathbb{E}_{x}[e^{-r(\tau\wedge\sigma_{k^{*}}\wedge
\tau_{n})}V(X_{\tau\wedge\sigma_{k^{*}}\wedge \tau_{n}})]\geq V(x)
\end{split}
\end{eqnarray}
By Lemma \ref{lem-ineq} we know,
\begin{equation}
\begin{split}
\mathbb{E}_{x}[e^{-r(\tau\wedge\sigma_{k^{*}}\wedge
\tau_{n})}V(X_{\tau\wedge\sigma_{k^{*}}\wedge \tau_{n}})]&\leq
\mathbb{E}_{x}[e^{-r(\tau\wedge\sigma_{k^{*}}\wedge
\tau_{n})}[((X_{\tau\wedge \tau_{n}}-K)^{+}+\delta)1_{\{\tau\wedge
\tau_{n}<\sigma_{k^{*}}\}}\\
& \quad \quad +(X_{\sigma_{k^{*}}}-K)^{+}1_{\{\sigma_{k^{*}}\leq
\tau\wedge \tau_{n}\}}]]
\end{split}
\end{equation}
Then, two applications of the bounded convergence theorem (while recalling the continuity of $V$) yields
\begin{equation}
\begin{split}
V(x)&\leq
\mathbb{E}_{x}[e^{-r(\tau\wedge\sigma_{k^{*}})}V(X_{\tau\wedge\sigma_{k^{*}}})]\\
&\leq
\mathbb{E}_{x}[e^{-r(\tau\wedge\sigma_{k^{*}})}[((X_{\tau}-K)^{+}+\delta)1_{\{\tau<\sigma_{k^{*}}\}}
+(X_{\sigma_{k^{*}}}-K)^{+}1_{\{\sigma_{k^{*}}\leq\tau\}}]]
\end{split}
\end{equation}
Hence,
\begin{equation}
\begin{split}
V(x)&\leq
\inf_{\tau\in\mathcal{S}_{0,\infty}}\mathbb{E}_{x}[e^{-r(\tau\wedge\sigma_{k^{*}})}[((X_{\tau}-K)^{+}+\delta)1_{\{\tau<\sigma_{k^{*}}\}}
+(X_{\sigma_{k^{*}}}-K)^{+}1_{\{\sigma_{k^{*}}\leq\tau\}}]]
\end{split}
\end{equation}
Thus, the opposite equality has been established and the following
relations hold.
\begin{equation}
\begin{split}
V(x)&= \inf_{\tau\in\mathcal{S}_{0,\infty}}\mathbb{E}_{x}
[e^{-r(\tau\wedge\sigma_{k^{*}})}[((X_{\tau}-K)^{+}+\delta)1_{\{\tau<\sigma_{k^{*}}\}}+(X_{\sigma_{k^{*}}}-K)^{+}1_{\{\sigma_{k^{*}}\leq
\tau\}}]]\\
&\leq
\sup_{\sigma\in\mathcal{S}_{0,\infty}}\inf_{\tau\in\mathcal{S}_{0,\infty}}
\mathbb{E}_{x}[e^{-r(\tau\wedge\sigma)}[((X_{\tau}-K)^{+}+\delta)1_{\{\tau<\sigma\}}+(X_{\sigma}-K)^{+}1_{\{\sigma\leq
\tau\}}]]\\
\end{split}
\end{equation}
This completes the justification and $V(x)=V^{*}(x)$ when
$x\in(h^{*},k^{*})$ as desired.\\
\indent Suppose $x\in(0,K)$.  Using the fact that $\mathcal{L}V-rV$ in $(0,K)$ and the optional sampling theorem, we know for each $n\geq 1$ and any $\sigma\in\mathcal{S}_{0,\infty}$,
\begin{eqnarray}
\begin{split}
\mathbb{E}_{x}[e^{-r(\tau_{[K,h^{*}]}\wedge
\sigma \wedge \tau_{n})}V(X_{\tau_{[K,h^{*}]}\wedge
\sigma \wedge \tau_{n}})]=V(x)
\end{split}
\end{eqnarray}
An application of the bounded convergence theorem (while recalling the continuity $V$) followed by Lemma \ref{lem-ineq} produces
\begin{equation}
\begin{split}
V(x)\geq \sup_{\sigma\in
\mathcal{S}_{0,\infty}}\mathbb{E}_{x}[e^{-r(\sigma \wedge
\tau_{[K,h^{*}]})}\{((X_{\tau_{[K,h^{*}]}}-K)^{+}+\delta)1_{\{\tau_{[K,h^{*}]}<\sigma\}}
+(X_{\sigma}-K)^{+}1_{\{\sigma\leq
\tau_{[K,h^{*}]}\}}\}]\
\end{split}
\end{equation}  
Thus, the first inequality in (\ref{pfstrat}) holds.  In addition, since $\tau_{[K,h^{*}]}<\sigma_{k^{*}}$, we have
\begin{equation}
\begin{split}
V(x)&=\mathbb{E}_{x}[e^{-r(\tau_{[K,h^{*}]}\wedge
\sigma_{k^{*}})}V(X_{\tau_{[K,h^{*}]}\wedge
\sigma_{k^{*}}})]\\
& \quad =\mathbb{E}_{x}[e^{-r(\sigma_{k^{*}} \wedge
\tau_{[K,h^{*}]})}\{((X_{\tau_{[K,h^{*}]}}-K)^{+}+\delta)1_{\{\tau_{[K,h^{*}]}<\sigma_{k^{*}}\}}
+(X_{\sigma_{k^{*}}}-K)^{+}1_{\{\sigma_{k^{*}}\leq
\tau_{[K,h^{*}]}\}}\}]
\end{split}
\end{equation}
which implies
\begin{equation}
\label{Vinf-02}
\begin{split}
V(x) &\geq
\inf_{\tau\in\mathcal{S}_{0,\infty}}\mathbb{E}_{x}[e^{-r(\sigma_{k^{*}}
\wedge \tau)}\{((X_{\tau}-K)^{+}+\delta)1_{\{\tau<\sigma_{k^{*}}\}}
+(X_{\sigma_{k^{*}}}-K)^{+}1_{\{\sigma_{k^{*}}\leq \tau\}}\}]
\end{split}
\end{equation}
The same argument used when $x\in(h^{*},k^{*})$ applies here to show that the opposite inequality in (\ref{Vinf-02}) holds.  Therefore, $V^{*}(x)=V(x)$ when $x\in(0,K)$. 

\indent Suppose $x\in[K,h^{*}]$.  Using (\ref{mart-01}) and the fact that $\tau_{[K,h^{*}]}=0$, for any stopping time $\sigma\in\mathcal{S}_{0,\infty}$,
\begin{equation}
\label{sup-01}
\begin{split}
V(x)&\geq \mathbb{E}_{x}[e^{-r(\sigma \wedge
\tau_{[K,h^{*}]})}\{((X_{\tau_{[K,h^{*}]}}-K)^{+}+\delta)1_{\{\tau_{[K,h^{*}]}<\sigma\}}
+(X_{\sigma}-K)^{+}1_{\{\sigma\leq
\tau_{[K,h^{*}]}\}}\}]
\end{split}
\end{equation}

Note that equality in (\ref{sup-01}) actually holds.  Now, taking the supremum over all stopping times in (\ref{sup-01}) yields the first inequality in (\ref{pfstrat}).
Again using (\ref{mart-01}) and $\tau_{[K,h^{*}]}=0$ yields,
\begin{equation}
\label{Vinf-03}
\begin{split}
V(x) &\geq
\inf_{\tau\in\mathcal{S}_{0,\infty}}\mathbb{E}_{x}[e^{-r(\sigma_{k^{*}}
\wedge \tau)}\{((X_{\tau}-K)^{+}+\delta)1_{\{\tau<\sigma_{k^{*}}\}}
+((X_{\sigma_{k^{*}}}-K)^{+})1_{\{\sigma_{k^{*}}\leq \tau\}}\}]
\end{split}
\end{equation}
The same argument used when $x\in(h^{*},k^{*})$ applies here to show
the opposite inequality in (\ref{Vinf-03}).  Thus,  $V(x)=V^{*}(x)$.  \\
\indent Finally, suppose $x\in[k^{*},\infty)$.  By Lemma \ref{lem-ineq}, we know
\begin{equation}
\begin{split}
&\mathbb{E}_{x}[e^{-r(\sigma \wedge
\tau_{[K,h^{*}]})}\{((X_{\tau_{[K,h^{*}]}}-K)^{+}+\delta)1_{\{\tau_{[K,h^{*}]}<\sigma\}}
+(X_{\sigma}-K)^{+}1_{\{\sigma\leq
\tau_{[K,h^{*}]}\}}\}] \\
\quad &\leq \mathbb{E}_{x}[e^{-r(\sigma \wedge \tau_{[K,h^{*}]}}V(X_{\sigma \wedge \tau_{[K,h^{*}]}})]
\end{split}
\end{equation} 
Using Lemma \ref{lemma-hk} and the optional sampling theorem for any $\sigma\in \mathcal{S}_{0,\infty}$ and $n\geq 1$, we have
\begin{equation}
\begin{split}
\mathbb{E}_{x}[e^{-r(\sigma \wedge \tau_{[K,h^{*}]}\wedge \tau_{n})}V(X_{\sigma \wedge \tau_{[K,h^{*}]}\wedge \tau_{n}})] &\leq V(x)
\end{split}
\end{equation}
By Fatou's Lemma and the continuity of $V$, we conclude
\begin{equation}
\begin{split}
&\mathbb{E}_{x}[e^{-r(\sigma \wedge
\tau_{[K,h^{*}]})}\{((X_{\tau_{[K,h^{*}]}}-K)^{+}+\delta)1_{\{\tau_{[K,h^{*}]}<\sigma\}}
+(X_{\sigma}-K)^{+}1_{\{\sigma\leq
\tau_{[K,h^{*}]}\}}\}] \\
\quad &\leq V(x)
\end{split}
\end{equation}
Taking the supremum over all stopping times $\sigma$ yields the first inequality in (\ref{pfstrat}).  Since $\sigma_{k^{*}}=0$, for any stopping time $\tau\in\mathcal{S}_{0,\infty}$,
\begin{equation}
\label{Vinf-04}
\begin{split}
V(x) &=\mathbb{E}_{x}[e^{-r(\sigma_{k^{*}}
\wedge \tau)}\{((X_{\tau}-K)^{+}+\delta)1_{\{\tau<\sigma_{k^{*}}\}}
+((X_{\sigma_{k^{*}}}-K)^{+})1_{\{\sigma_{k^{*}}\leq \tau\}}\}]
\end{split}
\end{equation}
Thus, we have
\begin{equation}
\label{Vinf-04}
\begin{split}
V(x) &=
\inf_{\tau\in\mathcal{S}_{0,\infty}}\mathbb{E}_{x}[e^{-r(\sigma_{k^{*}}
\wedge \tau)}\{((X_{\tau}-K)^{+}+\delta)1_{\{\tau<\sigma_{k^{*}}\}}
+((X_{\sigma_{k^{*}}}-K)^{+})1_{\{\sigma_{k^{*}}\leq \tau\}}\}]
\end{split}
\end{equation}
Thus, $V(x)=V^{*}(x)$.  This completes the proof.
\end{proof}

\begin{figure}[h!]
\centering \epsfig{file=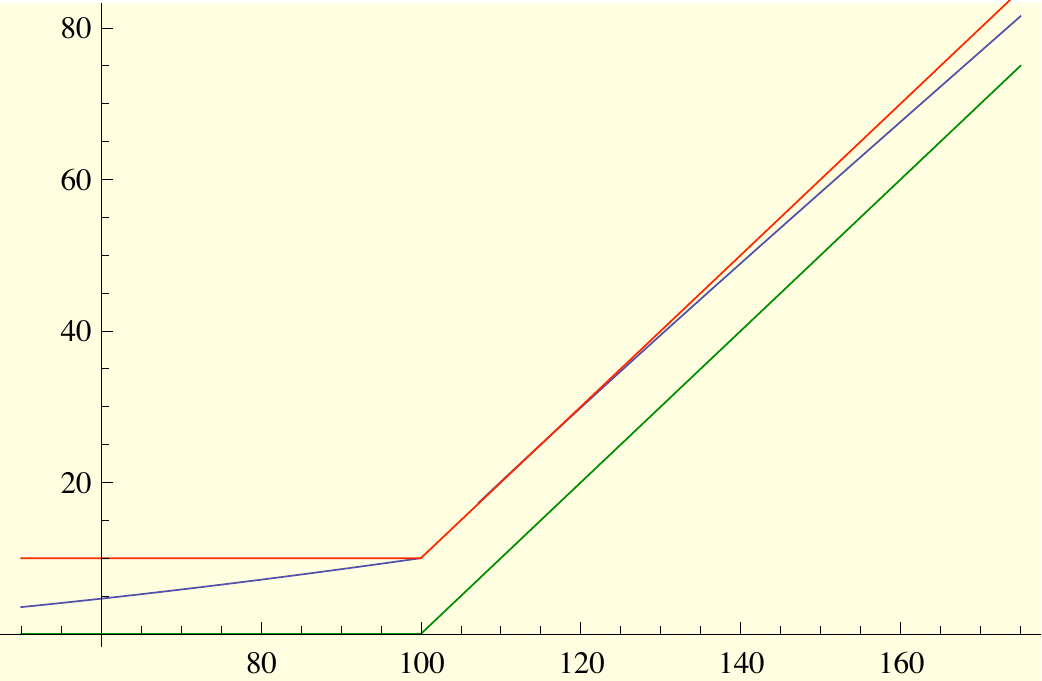,width=0.9\linewidth,clip=}
\caption[This figure displays the value function for the
$\delta$-penalty call option and the immediate exercise value
functions.]{This figure displays the value function for the
$\delta$-penalty call option and the immediate exercise value
functions on the interval $[50,175]$. Parameter values are:
$r=0.02$, $d=0.01$, $\delta=10$, $\sigma=0.2$, $h^{*}=107.50$,
$k^{*}=329.90$.}\label{PerCallRD01}
\end{figure}

\begin{figure}[h!]
\centering \epsfig{file=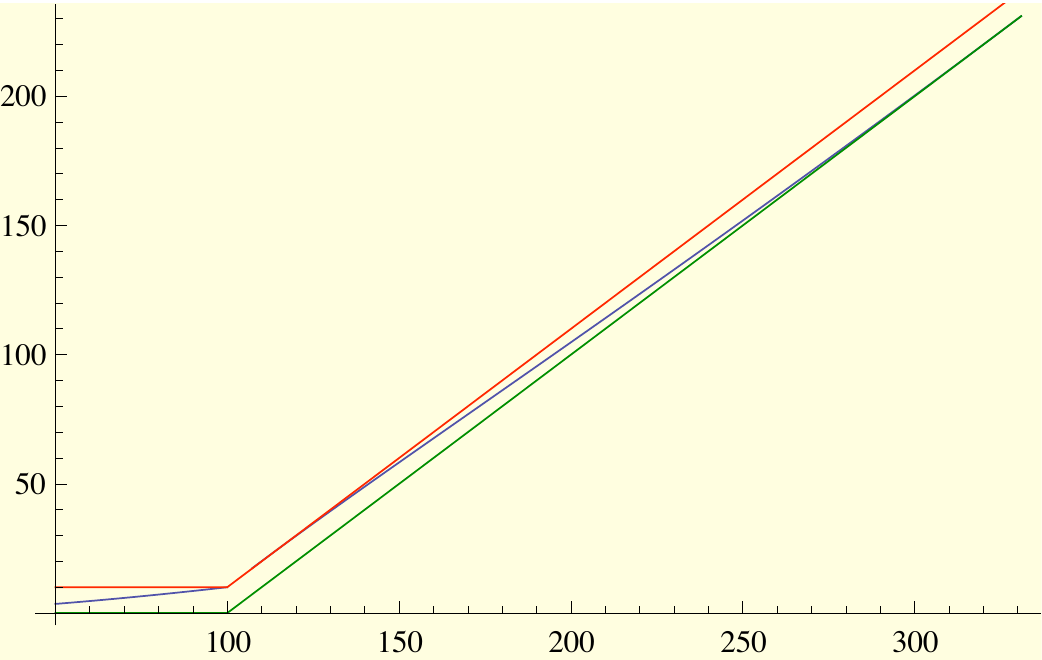,width=0.9\linewidth,clip=}
\caption[This figure displays the value function for the
$\delta$-penalty call option and the immediate exercise value
functions.]{This figure displays the value function for the
$\delta$-penalty call option and the immediate exercise value
functions on the interval $[50,330]$. Parameter values are:
$r=0.02$, $d=0.01$, $\delta=10$, $\sigma=0.2$, $h^{*}=107.50$,
$k^{*}=329.90$.}\label{PerCallRD02}
\end{figure}

\section{Numerical Results}
\label{CanNumRes} \indent  This section presents numerical results
pertaining to the $\delta$-penalty call option when $r>d$.  Recall,
when the interest rate exceeds the dividend yield, the price
function is not always a
convex function for all values of the underlying asset.\\
\indent Figure \ref{PerCallRD01} displays the value function for the
$\delta$-penalty call option and the immediate exercise value
functions on $[50,175]$ with parameter values: $r=0.02$, $d=0.01$,
$\sigma=0.2$, $K=100$, $\delta=10$.  We see that the value function
smoothly joins the upper immediate exercise value function at
$h^{*}=107.50$. Thus, the immediate cancellation region is the
interval $[100,107.50]$.  In addition, Figure \ref{PerCallRD02}
shows that the value function smoothly joins the lower immediate
exercise value function at $k^{*}=329.90$. Hence, the immediate
exercise region consists of the interval $[329.90,\infty)$.\\
\indent Our analysis in the previous section and the value function
featured in Figure \ref{PerCallRD01} highlight the fact that the
price of the $\delta$-penalty call option need not be a convex
function of the underlying asset even though the payoff is convex.
This result, though striking, is not unexpected from previous
analysis done on game-style options (see e.g. Ekstr\"{o}m (2006)).
Moreover, our results show that the $\delta$-penalty call option is
not necessarily non-decreasing in the volatility parameter.  Indeed,
Table \ref{tablePerpCall} shows that for asset values of $X=120,
130, 140, 150$, the $\delta$-penalty call option is decreasing in
volatility for model parameters $r=0.02$, $d=0.01$, $\delta=10$,
$K=100$.  Note this phenomenon occurs near non-convex pieces of the
value function.  Not surprisingly, this quality of the price
disappears as asset values approach $k^{*}$ and the value function
switches to being convex. Indeed, Table \ref{tablePerpCall}
indicates that prices are increasing in volatility when $X=280$ and
$X=290$. This numerical example highlights the close relationship
between the convexity of the price function and its monotonicity
with respect to the volatility
parameter.  \\
\indent  The price savings over a perpetual American Call option can
be substantial.  Since optimal cancellation occurs in an interval
with the strike as the left endpoint, we would expect the greatest
savings to occur close to this interval.  Indeed, we see from Table
\ref{tablePerpCall} that the cost savings to the investor of a
$\delta$-penalty call option is greatest at $X=120$ for any fixed
$\sigma$ value.  In fact, for $X=120$ and $\sigma=0.15$, the cost
savings of $26.5132$ represents nearly $89\%$ of the option value
and nearly $47\%$ of the regular American call option value.

\begin{table}
\[
\begin{tabular}{c}
\hline Table: Perpetual Cancellable Call, Perpetual American Call, Savings premia \\
\hline \multicolumn{1}{l}{
\begin{tabular}{l}
\begin{tabular}{lllll}
{\small Asset} & {\small Volatility} & {\small Canc.
Call} & {\small Amer. Call \ } & {\small Savings Premia \ } %
\end{tabular}
\\ \hline
\begin{tabular}{lllll}
& {\small \ \ \ \ 0.15 \ \ \ } & \ \ \ {\small 29.9499\ \ \ \ } & {\small %
56.4631} & \ {\small 26.5132}\\
{\small \ \ 120 } & {\small \ \ \ \ 0.20
 \ \ \ } & \ \ \ {\small
29.7883 \ \ } & {\small 64.3987} & \ {\small 34.6104
}\\
& {\small \ \ \ \ 0.25} & \ \ \ {\small 29.6394 \ \ } & {\small
71.4192 } & \ {\small 41.7798 } \\ &  &  &  &
\end{tabular}
\\
\begin{tabular}{lllll}
& {\small \ \ \ \ 0.15 \ \ \ } & \ \ \ {\small 39.5982
 \ \ \ } & {\small %
63.1082} & \ {\small 23.5100}\\
{\small \ \ 130  } & {\small \ \ \ \ 0.20 \ \ \ } & \ \ \ {\small 39.3874 \ \ } & {\small %
71.3509} & \ {\small 31.9635}\\
& {\small \ \ \ \ 0.25} & \ \ \ {\small 39.2417 \ \ } & {\small 78.6776 } & \ {\small 39.4359 }\\
&  &  &  &
\end{tabular}
\\
\begin{tabular}{lllll}
& {\small \ \ \ \ 0.15 \ \ \ } &  \ \ \ {\small 49.0096
 \ \ \ } &
{\small 69.9558} & {\small \ 20.9462}\\
{\small \ \ 140 } & {\small \ \ \ \ 0.20 \ \ \ } & \ \ \ {\small
48.8518 \ \ } &
{\small 78.4550} & {\small \ 29.6032}\\
& {\small \ \ \ \ 0.25} & \ \ \ {\small 48.7566 \ \ } & {\small 86.0539 } & \ {\small %
37.2973}\\
&  &  &  &
\end{tabular}
\\
\begin{tabular}{lllll}
& {\small \ \ \ \ 0.15 \ \ \ } & \ \ \ {\small 58.2716 \ \ \ } &
{\small 76.9971} & {\small \ 18.7255}\\
{\small \ \ 150 } & {\small \ \ \ \ 0.20 \ \ \ } & \ \ \ {\small 58.2283 \ \ } & {\small %
85.7032} & {\small \ 27.4749}\\
& {\small \ \ \ \ 0.25} & \ \ \ {\small 58.2134 \ \ } & {\small
93.5413 } & \
{\small 35.3279}\\
&  &  &  &
\end{tabular}
\\
\begin{tabular}{lllll}
& {\small  \ \ \ \ 0.15 \ \ \ } & \ \ \ {\small 180.1030 \ \ } &
{\small 183.3450} & {\small \ 3.2420}\\
{\small \ \ 280 } & {\small \ \ \ \ 0.20 \ \ \ } & \ \ \ {\small 180.7380 \ } & {\small %
190.6220} & {\small \ 9.8840}\\
& {\small \ \ \ \ 0.25} & \ \ \ {\small 181.4580 \ } & {\small
198.9720} &
{\small 17.5140}\\
&  &  &  &
\end{tabular}
\\
\begin{tabular}{lllll}
& {\small  \ \ \ \ 0.15 \ \ \ } & \ \ \ {\small 190.0100 \ \ } &
{\small 192.5100} & {\small \ 2.5000}\\
{\small \ \ 290 } & {\small \ \ \ \ 0.20 \ \ \ } & \ \ \ {\small
190.4730 \
} & {\small %
199.3850} & \ {\small  8.9120}\\
& {\small \ \ \ \ 0.25} & \ \ \ {\small 191.1390 \ } & {\small
207.5970} &
{\small 16.4580}%
\end{tabular}
\end{tabular}
}
\end{tabular}
\]
\caption[Table: Perpetual Cancellable Call Option, Perpetual Call
Option, Savings premia]{\noindent \underline{{\small Note}}{\small :
Columns 1 and 2 give the
underlying asset price }$X${\small \ and its return volatility }$\sigma $%
{\small . Columns 3 and 4 provide the $\delta$-penalty Call Option
price and the Perpetual American Call Option price. Column 5 is the
savings from purchasing a $\delta$-penalty call over a Perpetual
American call.  Parameter values are }$r=0.02$, $d=0.01$,
$\delta=10$, $K=100,$.  $\sigma=0.15 \ \Rightarrow h^{*}=115.0460, \
k^{*}=294.5790$.  $\sigma=0.20 \ \Rightarrow h^{*}=107.4860, \
k^{*}=329.8960$. $\sigma=0.25 \ \Rightarrow h^{*}=101.0210, \
k^{*}=365.7920$.}\label{tablePerpCall}
\end{table}

\section{Conclusion}
\label{CanConclusion}

The above discussion presents the valuation of the perpetual
$\delta$-penalty call option.  This analysis follows the work done
in Kyprianou (2004) with respect to the perpetual $\delta$-penalty
put option.  We find that the solution to the problem differs
considerably depending on the relative values of the interest rate
and dividend yield for the underlying asset. Specifically, when
$r\leq d$, analogous arguments to Kyprianou (2004) identify the
explicit solution to the valuation problem. Namely, the value of the
claim corresponds to its price under the policy of exercising at the
first time the underlying asset reaches an optimally chosen value
$k^{*}$ and under the policy of terminating the contract when the
asset value first reaches the strike price $K$.  In addition, the
value function is a convex function of the underlying asset price.
When $r>d$, the optimal cancellation region no longer is the
singleton $\{K\}$ in general.  Instead, it consists of an interval
of the form $[K,h^{*}]$; where $h^{*}$ must be determined as
part of the solution.  We show
that $h^{*}$ and $k^{*}$ respects two natural bounds. Namely, the optimal termination point satisfies
$h^{*}\leq \frac{r(K-\delta)}{d}$ and the optimal exercise
point satisfies $k^{*}\geq \frac{r}{d}K$.  In addition, smooth-pasting
holds both at the holder's optimal exercise boundary value $k^{*}$
and the writer's cancellation value $h^{*}$. This
striking result implies that the price is not a convex function for
all values of the underlying asset. Further, numerical solutions for
the valuation problem show that the value function is not
necessarily non-decreasing in the volatility parameter.  This phenomenon
directly relates to the existence of non-convex pieces of the value
function. Finally, we observe significant price savings over the
perpetual call option.  This savings might be especially appealing
to purchasers seeking a call option position who are willing to
assume the risk of cancellation.

\section{Appendix}
\label{Appendix}
\subsection{Appendix for Section \ref{val-dr}}
\begin{prpn}
$\delta(k)$ is an increasing function.
\end{prpn}

\begin{proof}
Solving equation (\ref{FOC1}) for $\delta$, we find 
\begin{equation}
\begin{split}
\delta(k)=-\frac{\left(\frac{K}{k}\right)^{-\frac{\lambda+\kappa\sigma}{\sigma}}(k(1-\left(\frac{k}{K}\right)^{-\frac{2\lambda}{\sigma}})\sigma-(k-K)(\lambda-\kappa\sigma+\left(\frac{K}{k}\right)^{\frac{2\lambda}{\sigma}}(\lambda+\kappa\sigma))}{2\lambda}
\end{split}
\end{equation}
Taking a derivative and simplifying yields

\begin{equation}
\begin{split}
\delta'(k)&=\frac{1}{2k\lambda\sigma}\left(\frac{k}{K}\right)^{-\frac{2\lambda}{\sigma}}\left(\frac{K}{k}\right)^{-\frac{\lambda+\kappa\sigma}{\sigma}}\times\Bigg(\left(\frac{k}{K}\right)^{\frac{2\lambda}{\sigma}}K\left(-1+\left(\frac{K}{k}\right)^{\frac{2\lambda}{\sigma}}\right)(\lambda^{2}-\kappa^{2}\sigma^{2})\\
& \quad \quad - k(\lambda-(1+\kappa)\sigma)\left(\sigma+\left(\frac{k}{K}\right)^{\frac{2\lambda}{\sigma}}\left(-\lambda -(1+\kappa)\sigma+\left(\frac{K}{k}\right)^{\frac{2\lambda}{\sigma}}(\lambda+\kappa\sigma)\right)\right)\Bigg)
\end{split}
\end{equation}
We now show that the derivative is non-negative.  We can neglect the first three factors of the above derivative since they are all positive.  From this point, we will utilize the substitution $y:=\frac{k}{K}$ to ease notation.  At this point, we want to show

\begin{equation}
\begin{split}
&\left(-1+\left(\frac{1}{y}\right)^{\frac{2\lambda}{\sigma}}\right)y^{\frac{2\lambda}{\sigma}}(\lambda^{2}-\kappa^{2}\sigma^{2})- y(\lambda-(1+\kappa)\sigma)\\
& \quad \quad \quad \times \left(\sigma+y^{\frac{2\lambda}{\sigma}}\left(-\lambda -(1+\kappa)\sigma+\left(\frac{1}{y}\right)^{\frac{2\lambda}{\sigma}}(\lambda+\kappa\sigma)\right)\right) \geq 0
\end{split}
\end{equation}

This is equivalent to showing 

\begin{equation}
\begin{split}
&\frac{\left(-1+\left(\frac{1}{y}\right)^{\frac{2\lambda}{\sigma}}\right)y^{-1+\frac{2\lambda}{\sigma}}(\lambda^{2}-\kappa^{2}\sigma^{2})}{(\lambda-(1+\kappa)\sigma)
\times \left(\sigma+y^{\frac{2\lambda}{\sigma}}\left(-\lambda -(1+\kappa)\sigma+\left(\frac{1}{y}\right)^{\frac{2\lambda}{\sigma}}(\lambda+\kappa\sigma)\right)\right)} \leq 1
\end{split}
\end{equation}

Since $y\geq 1$, it suffices to show 

\begin{equation}
\begin{split}
&\frac{\left(-1+\left(\frac{1}{y}\right)^{\frac{2\lambda}{\sigma}}\right)y^{\frac{2\lambda}{\sigma}}(\lambda^{2}-\kappa^{2}\sigma^{2})}{(\lambda-(1+\kappa)\sigma)
\times \left(\sigma+y^{\frac{2\lambda}{\sigma}}\left(-\lambda -(1+\kappa)\sigma+\left(\frac{1}{y}\right)^{\frac{2\lambda}{\sigma}}(\lambda+\kappa\sigma)\right)\right)} \leq 1
\end{split}
\end{equation}

Or equivalently show,

\begin{equation}
\begin{split}
&\frac{\left(1-\left(\frac{1}{y}\right)^{\frac{2\lambda}{\sigma}}\right)y^{\frac{2\lambda}{\sigma}}(\lambda^{2}-\kappa^{2}\sigma^{2})}{(\lambda-(1+\kappa)\sigma)
\times \left(-\sigma+y^{\frac{2\lambda}{\sigma}}\left(\lambda +(1+\kappa)\sigma-\left(\frac{1}{y}\right)^{\frac{2\lambda}{\sigma}}(\lambda+\kappa\sigma)\right)\right)} \leq 1
\end{split}
\end{equation}

Algebraic manipulations of the left-hand side produce

\begin{equation}
\begin{split}
&\frac{\left(y^{\frac{2\lambda}{\sigma}}-1\right)(\lambda^{2}-\kappa^{2}\sigma^{2})}{(\lambda-(1+\kappa)\sigma)
\times(\lambda+\kappa\sigma+\sigma)\times \left(y^{\frac{2\lambda}{\sigma}}-1\right)}\\
&\Leftrightarrow\frac{\lambda^{2}-\kappa^{2}\sigma^{2}}{(\lambda-(1+\kappa)\sigma)
\times(\lambda+\kappa\sigma+\sigma)}\\
&\Leftrightarrow \frac{\lambda^{2}-\kappa^{2}\sigma^{2}}{\lambda^{2}-\kappa^{2}\sigma^{2}-\sigma^{2}(1+2\kappa)}
\end{split}
\end{equation}

Since $\kappa:=\frac{r-d-\frac{\sigma^{2}}{2}}{\sigma^{2}}$ and $r\leq d$, it follows that $-\sigma(1+2\kappa)\geq 0$.  Thus, the left-hand side is less than or equal to $1$ and the proof is complete. 
\end{proof}
\begin{prpn}
The function $f(x)$ satisfies 
\begin{equation}
f(x)\geq 0, x \in [K,b]
\end{equation}
\end{prpn}
\begin{proof}
Algebraic simplification yields 
\begin{equation}
\begin{split}
f(x)&=
\frac{1}{2}\left(\frac{x}{K}\right)^{\frac{\lambda}{\sigma}+\kappa}\\
&\left(2\left(\frac{x}{K}\right)^{-\frac{2\lambda}{\sigma}}(x-K)+\frac{x\left(1-\left(\frac{x}{K}\right)^{-\frac{2\lambda}{\sigma}}\right)\sigma
-(x-K)\left(\lambda-\kappa\sigma+\left(\frac{x}{K}\right)^{-\frac{2\lambda}{\sigma}}
(\lambda+\kappa\sigma)\right)}{\lambda}\right)
\end{split}
\end{equation}
Since the first factor in the above expression is positive we can
discard this from our analysis.  Now, multiplying throughout by
$\lambda$ leads us to showing the following condition holds for $x
\in [K,b]$.
\begin{equation}
\begin{split}
&2\left(\frac{x}{K}\right)^{-\frac{2\lambda}{\sigma}}(x-K)\lambda+x\left(1-\left(\frac{x}{K}\right)^{-\frac{2\lambda}{\sigma}}\right)\sigma
\geq
(x-K)\left(\lambda-\kappa\sigma+\left(\frac{x}{K}\right)^{-\frac{2\lambda}{\sigma}}
(\lambda+\kappa\sigma)\right)
\end{split}
\end{equation}
In order to further simplify this inequality, we make the
substitution $y:=\frac{x}{K}$.  This yields the following inequality
\begin{equation}
\begin{split}
&2y^{-\frac{2\lambda}{\sigma}}(y-1)K\lambda+y\left(1-y^{-\frac{2\lambda}{\sigma}}\right)K\sigma
\geq (y-1)K\left(\lambda-\kappa\sigma+y^{-\frac{2\lambda}{\sigma}}
(\lambda+\kappa\sigma)\right)
\end{split}
\end{equation}
for $1\leq y \leq \frac{b}{K}$.  Recall, $y=1$ and $y=\frac{b}{K}$
both satisfy this inequality.  As a result, let us consider
$1<y<\frac{b}{K}=
\frac{\lambda-\kappa\sigma}{\lambda-\kappa\sigma-\sigma}$.  The
condition now can be reduced to showing
\begin{equation}
\begin{split}
&2y^{-\frac{2\lambda}{\sigma}}\lambda+\frac{y}{y-1}\left(1-y^{-\frac{2\lambda}{\sigma}}\right)\sigma
\geq \lambda-\kappa\sigma+y^{-\frac{2\lambda}{\sigma}}
(\lambda+\kappa\sigma)
\end{split}
\end{equation}
for $1<y<\frac{\lambda-\kappa\sigma}{\lambda-\kappa\sigma-\sigma}$.
Straightforward algebra shows the following sequence of relations
can all be deduced from each other.
\begin{equation}
\begin{split}
y^{-\frac{2\lambda}{\sigma}}(2\lambda-\lambda-\kappa\sigma)
+\left(\frac{y}{y-1}\right)(1-y^{-\frac{2\lambda}{\sigma}})\sigma
&\geq \lambda-\kappa\sigma\\
y^{-\frac{2\lambda}{\sigma}}(\lambda-\kappa\sigma)
+\left(\frac{y}{y-1}\right)(1-y^{-\frac{2\lambda}{\sigma}})\sigma
&\geq \lambda-\kappa\sigma\\
\left(\frac{y}{y-1}\right)(1-y^{-\frac{2\lambda}{\sigma}})\sigma
&\geq(\lambda-\kappa\sigma)(1-y^{-\frac{2\lambda}{\sigma}})\\
\frac{y}{y-1}\sigma &\geq \lambda-\kappa\sigma\\
y\sigma &\geq \lambda y -\kappa\sigma y -\lambda+\kappa\sigma\\
\lambda-\kappa\sigma &\geq y(\lambda-\kappa\sigma-\sigma)\\
y &\leq \frac{\lambda-\kappa\sigma}{\lambda-\kappa\sigma-\sigma}
\end{split}
\end{equation}
Notice the last inequality is precisely the case under
consideration.  Thus, all of the above inequalities are true and
we have shown $f(x)\geq 0$ for $x\in[K,b]$. 
\end{proof}

\subsection{Appendix for Section \ref{val-rd}}
 \begin{lem}
 A pair $(h^{*},k^{*})$ solving the equations (\ref{opt-hk}) with $K<h^{*}<k^{*}$ satisfy $h^{*}<\frac{r(K-\delta)}{d}$ and $k^{*}>\frac{r}{d}K$.
 \end{lem}
 
 \begin{proof}
 The following argument is inspired by the proof of Theorem 4.3 in Alvarez (2008).  Let $g_{1}(x):=(x-K)^{+}$ and $g_{2}(x):=(x-K)^{+}+\delta$.  First, note that for $K<x<\frac{r}{d}K$, $(\mathcal{L}g_{1}-r g_{1})(x)>0$; for $x=\frac{r}{d}K$, $(\mathcal{L}g_{1}-r g_{1})(x)=0$; for $x>\frac{r}{d}K$, $(\mathcal{L}g_{1}-r g_{1})(x)<0$.  Second, note that for $K<x<\frac{r(K-\delta)}{d}$, $(\mathcal{L}g_{2}-r g_{2})(x)>0$; for $x=\frac{r(K-\delta)}{d}$, $(\mathcal{L}g_{2}-r g_{2})(x)=0$; for $x>\frac{r(K-\delta)}{d}$, $(\mathcal{L}g_{2}-r g_{2})(x)<0$.  Third, notice 
 \begin{eqnarray}
 \label{deriv-eq}
 \begin{split}
 \frac{\mathrm{d}}{\mathrm{d}x}\left(    \frac{g'_{2}(x)}{S'(x)}\hat{\varphi}_{k}(x)-\frac{\hat{\varphi}'_{k}(x)}{S'(x)}g_{2}(x) \right)&=(\mathcal{L}g_{2}-r g_{2})(x)\hat{\varphi}_{k}(x)m'(x)\\
 \frac{\mathrm{d}}{\mathrm{d}x}\left(    \frac{g'_{1}(x)}{S'(x)}\hat{\psi}_{h}(x)-\frac{\hat{\psi}'_{h}(x)}{S'(x)}g_{1}(x) \right)&=(\mathcal{L}g_{1}-r g_{1})(x)\hat{\psi}_{h}(x)m'(x)
 \end{split}
 \end{eqnarray}
 Thus, equations (\ref{opt-hk}) can be re-expressed as 
 
 \begin{eqnarray}
 \begin{split}
 \label{opt-Alv}
 B^{-1}\int_{h^{*}}^{k^{*}}(\mathcal{L}g_{2}-r g_{2})(x)\psi(k^{*})\hat{\varphi}_{k^{*}}(x)m'(x)\mathrm{d}x &=g_{2}(k^{*})-g_{1}(k^{*})\\
 B^{-1}\int_{h^{*}}^{k^{*}}(\mathcal{L}g_{1}-r g_{1})(x)\varphi(h^{*})\hat{\psi}_{h^{*}}(x)m'(x)\mathrm{d}x &=g_{1}(h^{*})-g_{2}(h^{*})
 \end{split}
 \end{eqnarray}
 
 Now consider, for any fixed $k>K$, the function
 
 \begin{equation}
 \begin{split}
 L_{1}(h):=\frac{B}{\psi(k)}(g_{2}(k)-g_{1}(k))-\int_{h}^{k}(\mathcal{L}g_{2}-r g_{2})(x)\hat{\varphi}_{k}(x)m'(x)\mathrm{d}x 
 \end{split}
 \end{equation} 
 Notice $L_{1}(k)>0$ and $L_{1}(h)$ is increasing on $\left(K,\frac{r(K-\delta)}{d} \right)$, and decreasing on $\left(\frac{r(K-\delta)}{d},\infty \right)$.  Thus, if a root $h^{*}_{k}\in(K,k)$ satisfying $L_{1}(h^{*}_{k})=0$ exists, it must be on the set $(K,\frac{r(K-\delta)}{d})$.  Similarly, consider for any fixed $h>K$, the function 
 \begin{equation}
 \begin{split}
 L_{2}(k):=\frac{B}{\varphi(h)}(g_{1}(h)-g_{2}(h))-\int_{h}^{k}(\mathcal{L}g_{1}-r g_{1})(x)\hat{\psi}_{h}(x)m'(x)\mathrm{d}x 
 \end{split}
 \end{equation} 
 Notice $L_{2}(h)<0$ and $L_{2}(k)$ is decreasing on $\left(K, \frac{r}{d}K \right)$, and increasing on $\left(\frac{r}{d}K,\infty \right)$.  Hence, if a root $k^{*}_{h} \in (h,\infty)$ satisfying the condition $L_{2}(k^{*}_{h})=0$ exists, then it has to be on the set $\left(\frac{r}{d}K,\infty \right)$.  
 \end{proof}

\begin{lem}
\label{lem-ineq}
The value function $V(x)$ as defined in Theorem \ref{thm-rd} satisfies \[(x-K)^{+} \leq V(x) \leq (x-K)^{+}+\delta\]  
\end{lem}

\begin{proof}
In order to complete the proof, we only need to consider the case
when $x\in(h^{*},k^{*})$.  Indeed, notice $0 \leq V(x)\leq \delta$ when $x\in(0,K)$.  The following argument is inspired by the proof of Theorem 4.3 in Alvarez (2008).  Define for $g_{1}:=(x-K)^{+}$ and $g_{2}:=(x-K)^{+}+\delta$ the following functions,
\begin{equation}
\begin{split}
\triangle_{1}&:=V(x)-g_{1}(x)=g_{2}(h^{*})\frac{\hat{\varphi}_{k^{*}}(x)}{\hat{\varphi}_{k^{*}}(h^{*})}+g_{1}(k^{*})\frac{\hat{\psi}_{h^{*}}(x)}{\hat{\psi}_{h^{*}}(k^{*})}-g_{1}(x)\\
\triangle_{2}&:=V(x)-g_{2}(x)=g_{2}(h^{*})\frac{\hat{\varphi}_{k^{*}}(x)}{\hat{\varphi}_{k^{*}}(h^{*})}+g_{1}(k^{*})\frac{\hat{\psi}_{h^{*}}(x)}{\hat{\psi}_{h^{*}}(k^{*})}-g_{2}(x)
\end{split}
\end{equation}

Now, by our construction, continuity and smooth-pasting hold at $h^{*}$, $k^{*}$.  Thus,  $\triangle_{1}(k^{*})=\triangle_{1}'(k^{*})=0$ and $\triangle_{2}(h^{*})=\triangle_{2}'(h^{*})=0$.  Standard differentiation yields

\begin{equation}
\begin{split}
\frac{\mathrm{d}}{\mathrm{d}x}\left(\frac{\triangle_{1}(x)}{\hat{\varphi}_{k^{*}}(x)}\right)&=\frac{S'(x)}{\hat{\varphi}^{2}_{k^{*}}(x)}\left(\frac{Bg_{1}(k^{*})}{\psi(k^{*})}
-\frac{g_{2}'(x)}{S'(x)}\hat{\varphi}_{k^{*}}(x)+\frac{\hat{\varphi}'_{k^{*}}(x)}{S'(x)}g_{2}(x)\right) \\
\frac{\mathrm{d}}{\mathrm{d}x}\left(\frac{\triangle_{2}(x)}{\hat{\psi}_{h^{*}}(x)}\right)&=\frac{S'(x)}{\hat{\psi}^{2}_{h^{*}}(x)}\left(\frac{-Bg_{2}(h^{*})}{\varphi(h^{*})}
-\frac{g_{1}'(x)}{S'(x)}\hat{\psi}_{h^{*}}(x)+\frac{\hat{\psi}'_{h^{*}}(x)}{S'(x)}g_{1}(x)\right)
\end{split}
\end{equation}

Invoking equations (\ref{opt-hk}) and using the observations (\ref{deriv-eq}), we have

\begin{equation}
\begin{split}
\frac{\mathrm{d}}{\mathrm{d}x}\left(\frac{\triangle_{1}(x)}{\hat{\varphi}_{k^{*}}(x)}\right)&=-\frac{S'(x)}{\hat{\varphi}^{2}_{k^{*}}(x)}\int_{h^{*}}^{x}\hat{\varphi}_{k^{*}}(t)(\mathcal{L}g_{2}-r g_{2})(t)m'(t)\mathrm{d}t<0\\
\frac{\mathrm{d}}{\mathrm{d}x}\left(\frac{\triangle_{2}(x)}{\hat{\psi}_{h^{*}}(x)}\right)&=\frac{S'(x)}{\hat{\psi}^{2}_{h^{*}}(x)}\int^{k^{*}}_{x}\hat{\psi}_{h^{*}}(t)(\mathcal{L}g_{1}-r g_{1})(t)m'(t)\mathrm{d}t<0
\end{split}
\end{equation}
since $x\in(h^{*},k^{*})$ and $h^{*}<\frac{r(K-\delta)}{d}$ and $k^{*}>\frac{r}{d}K$.  Thus, have that $\triangle_{1}(x)\geq \triangle_{1}(k^{*})=0$ and $\triangle_{2}(x)\leq \triangle_{2}(h^{*})=0$ for all $x\in(h^{*},k^{*})$.  Hence, $g_{1}(x)\leq V(x) \leq g_{2}(x)$ for $x\in(h^{*},k^{*})$.  
\end{proof}

\end{document}